\documentclass{ocg}

\usepackage{amsmath, amssymb}
\usepackage{float}
\usepackage[bottom]{footmisc}
\usepackage{comment}
\usepackage{enumerate}
\usepackage{cite}
\usepackage{graphicx}
\usepackage{longtable}
\usepackage{wasysym}
\usepackage{multirow}
\usepackage{subfigure}
\usepackage{multirow}
\newcommand{\dollar}[0]{\$}

\newtheorem{fact}[theorem]{Fact}

\begin{document}

\title{LANGUAGE CLASSES ASSOCIATED WITH AUTOMATA OVER 
MATRIX GROUPS} 
\author[e]{\"{O}zlem Salehi}
\thanks[e]{\"{O}zlem Salehi is partially supported by T\"{U}B\.{I}TAK 
(Scientific and Technological Research Council of Turkey).} 
\thanks[z]{Flavio D'Alessandro acknowledges the support  of 7th FP 
T\"{U}B\.{I}TAK/Marie-Curie Co-Funded Brain Circulation Scheme, 2236.}
\author[z,c]{Flavio D'Alessandro}
\author[e]{Ahmet Celal Cem Say}

\address[e]{Bo\v{g}azi\c{c}i University, Department of Computer Engineering, \\
Bebek 34342 \.{I}stanbul, Turkey\\
\email{\{ozlem.salehi,say\}@boun.edu.tr}}
\address[z]{Bo\u gazi\c ci University, 
Department of Mathematics, \\ Bebek 34342, \.{I}stanbul, Turkey
}
\address[c]{Universit\`a di Roma ``La Sapienza'', Dipartimento di Matematica, \\
Piazzale Aldo Moro 2, 00185 Roma,  Italy \\
\email{dalessan@mat.uniroma1.it}
}

\maketitle

\begin{abstract}
We investigate the language classes recognized by group automata over matrix 
groups. We present a summary of the results obtained so far together with a number 
of new results. We look at the computational power of time-bounded group automata 
where the group under consideration has polynomial growth.
\end{abstract}

\section{Introduction}

Many extensions of the classical finite automaton model have been examined. 
One such variant is the group automaton (finite automaton over groups), which 
is a nondeterministic finite automaton equipped with a register which holds
an element from a group. The register is initialized to the identity
element of the group, and a computation is deemed successful if the
register is equal to the identity element at the end of the
computation after being multiplied at every step. This setup
generalizes various models such as nondeterministic blind multicounter
automata \cite{FMR67}, and finite automata with multiplication \cite{ISK76}.

Group automata were defined explicitly for the first time in the paper
\cite{MS97}. The theory of group automata has been essentially
developed in the case of free groups \cite{DM00,Co05,Ka09}, and in the
case of free Abelian groups \cite{EO04,EKO08}, where strong theorems
allow to characterize the power of such models and the combinatorial
properties of the languages recognized by these automata. The
connection between the word problem of a group and the class of
languages recognized by the associated automaton has been essential
while analysing group automata. Other papers which deal with group
automata and the word problems of groups include \cite{Ka06,CEO06}.

Our aim in this paper is to provide an overview of the languages
recognized by finite automata over matrix groups. Even in the case of
groups of matrices of low dimension that are not of the type mentioned
above, the study of group automata becomes quickly nontrivial, and
there are remarkable classes of linear groups for which little is
known about the automaton models that they define. We present some new
results about the classes of languages recognized by finite automata
over matrix groups. We also introduce the notion of time complexity
for group automata and use this to prove an impossibility result
about word problems of groups with exponential growth.

Section \ref{Section: pre} contains definitions and introduces
notation that will be used throughout the paper. In Section
\ref{Section: results}, we focus on matrix groups with integer and
rational entries. For the case of $ 2 \times 2 $ matrices, we prove
that the corresponding group automata for rational matrix groups are
more powerful than the corresponding group automata for integer 
matrix groups. We explore finite automata over some special
matrix groups such as the discrete Heisenberg
group and the Baumslag Solitar group. In Section \ref{Section:
time}, we consider group automata operating in bounded time and prove
that the word problem of a group with exponential growth cannot be
recognized by a finite automaton over a group with polynomial growth in
polynomial time, whereas in the succeeding section we summarize 
various results from the literature and some additional new related 
results. Section \ref{Section: open} lists some open questions.

\section{Preliminaries} \label{Section: pre}

\subsection{Notation and Terminology}

The following notation will be used throughout the paper: $Q$ is the set of
states, $q_0 \in Q$ denotes the initial state, $Q_a \subseteq Q$ denotes the
set of accepting states, and $\Sigma$ is the input alphabet.

An input string $w$ is placed between two endmarker symbols
on an infinite tape in the form $\cent w\dollar$. By $ w^r $, we represent the 
reverse of the string $ w $. The length of $ w $ is denoted by $ |w| $. 

$ \mathsf{REG} $, $ \mathsf{CF} $, and $ \mathsf{RE} $ denote the family of 
regular languages, context-free languages, and recursively enumerable languages,
respectively.

We assume a familiarity with some basic notions from algebra and group theory 
(see {\cite{Fr03},\cite{LS77}} for references on this topic). For a finitely generated 
group $ G $ and a set $ A $ of generators, the word problem language of $ G $ is 
the language $ W(G,A) $ over $ A $ which consists of all words that represent the 
identity element of $ G $. Most of the time, the statements about word problem 
are independent of the generating set and in these cases the word problem 
language is denoted by $ W(G) $.

\subsection{Group Automaton}

Group automata first appear explicitly in the paper 
\textit{The accepting power of finite automata over groups} by Mitrana and Stiebe 
\cite{MS97} under the name of extended finite automaton. 
The definition is formally given as follows.

Let $ K = (M,\circ, e )$ be a group under the operation denoted by 
$ \circ $ with the neutral element denoted by $ e $. 
An \textit{extended finite automaton} over the group $ K = (M,\circ, e)$ is a
6-tuple
\[ \mathcal{E} = (Q, \Sigma,K,\delta, q_0,Q_a) \]
where  the transition function $\delta$ is defined as
\[\delta: Q \times (\Sigma \cup \{\varepsilon\}) \rightarrow \mathbb{P}(Q\times M).\] 
$ \delta(q,\sigma) \ni (q',m) $ means that when $\mathcal{E}$ reads the
symbol (or empty string) $\sigma \in \Sigma \cup \{\varepsilon\}$
in state $q$, it will move to state $q'$, and write $ x\circ m $ in the register, 
where $ x $ is the old content of the register. The initial value of the register is the 
neutral element $ e $ of the group $ K $. The string is accepted if, after
completely reading the string,  $\mathcal{E}$ enters an accept state with the 
content of the register being equal to the neutral element of $ K $. 

We will prefer using the name group automaton ($ G $-automaton) instead of 
extended finite automaton over group $ G $.
 
The class of languages recognized by $ G $-automaton will be denoted as 
$ \mathfrak{L}(G) $.

\section{Matrix Groups and Associated Language Classes}
\label{Section: results}

In this section, we are going to prove some new results about the
classes of languages recognized by finite automata over various groups,
focusing on linear groups.

We will denote the \textit{free group} over $r$ generators by $ \mathbf{F}_r $.
Note that $ \mathbf{F}_0 $ is the trivial group, and $ \mathbf{F}_1 $
is isomorphic to $ \mathbb{Z} $, the additive group of integers. The
class of regular languages is characterized as the set of languages
recognized by finite automata over the trivial group $\mathbf{F}_0 $
in  \cite{DM00}.

We will denote by $ \mathbb{Z}^k $ the additive group of integer
vectors of dimension $k$. This group is
isomorphic to the \textit{free Abelian group of rank $ k$}, and  $
\mathbb{Z}^k $-automata are equivalent to nondeterministic blind
$k$-counter automata \cite{Gr78}. We denote by $ \mathbb{Q}^+$ the
multiplicative group of positive rational numbers, which is isomorphic
to a free Abelian group of infinite rank. A $ \mathbb{Q}^+$-automaton
is also equivalent to a nondeterministic finite automaton with
multiplication without equality (1NFAMW) of Ibarra et al.
\cite{ISK76}.

A characterization of context-free languages by group automata was
first stated by Dassow and Mitrana \cite{DM00}, and proven in \cite{Co05}. 
Let us note that $ \mathbf{F}_2 $ contains any free group
of rank $ n \geq 2 $ \cite{LS77}.

\begin{fact}\textup{\cite{DM00, Co05, Ka06}}\label{fact: cf}
	$\mathfrak{L}(\mathbf{F}_2)$ is the family of context-free languages.
\end{fact}

We denote by $GL(2,\mathbb{Z})$ the general linear group of degree two
over the field of integers, that is the
group of $ 2\times 2 $ invertible matrices with integer entries. Note
that these matrices have determinant $\pm 1$. Restricting the matrices
in $GL(2,\mathbb{Z})$ to those that have determinant 1, we obtain the
special linear group of degree two over the field of integers,
$SL(2,\mathbb{Z})$.

Let $ \mathbf{G} $ be the group generated by the  matrices
\[
M_{a}=
\left [
\begin{array}{cc}
1&2\\
0&1\\
\end{array}
\right ]~~~\mbox{and}~~~
M_{b}=
\left [
\begin{array}{cc}
1&0\\
2&1\\
\end{array}
\right ].
\]
There exists an isomorphism $ \varphi $ from $\mathbf{F}_2 $ onto
$\mathbf{G} $ by \cite{KM79}. Note that $ M_a $ and $ M_b $ are
integer matrices with determinant 1, which proves that $ \mathbf{F}_2
$ is a subgroup of $ SL(2,\mathbb{Z}) $.

Now the question is whether $ \mathfrak{L}(GL(2,\mathbb{Z}))$ and $
\mathfrak{L}(SL(2,\mathbb{Z}))$ correspond to larger classes of
languages than the class of context-free languages. We are going to
use the following fact to prove that the answer is negative.

\begin{fact}\textup{\cite{Co05}}  \label{fact: finite}
	Suppose $G$ is a finitely generated group and $H$ is a subgroup of
	finite index. Then $\mathfrak{L}(G) = \mathfrak{L}(H)$.
\end{fact}

Now we are ready to state our theorem.

\begin{theorem}\label{theorem:gl}
$\mathfrak{L}(SL(2,\mathbb{Z})) = \mathfrak{L}(GL(2,\mathbb{Z})) = \mathsf{CF} $. 
\end{theorem}
\begin{proof} We are going to use Fact  \ref{fact: finite} to prove the result. 
Since $ SL(2,\mathbb{Z}) $ has index 2 in $ GL(2,\mathbb{Z})$
and $GL(2,\mathbb{Z})$ is finitely generated, $\mathfrak{L}(GL(2,\mathbb{Z})) =
\mathfrak{L}(SL(2,\mathbb{Z}))$. Since $\mathbf{F}_2$ has index 12 in 
$SL(2,\mathbb{Z})$ \cite{BO08} and $SL(2,\mathbb{Z})$ is finitely generated, $
\mathfrak{L}(SL(2,\mathbb{Z})) = \mathfrak{L}(\mathbf{F}_2)$ which is equal to 
the family of context-free languages by Fact \ref{fact: cf}.
\end{proof}

Let us now investigate the group $SL(3,\mathbb{Z})$, the group of 
$3 \times 3 $ integer matrices with determinant~$1$.

We start by looking at an important subgroup of $SL(3,\mathbb{Z})$,
the discrete Heisenberg group. The discrete Heisenberg group
$\mathbf{H}$ is defined as $\langle a,b | ab=bac,ac=ca,bc=cb
\rangle$
where
$c=a^{-1}b^{-1}ab$ is the commutator of $a$ and $b$.

$$ a=\left [
\begin{array}{ccc}
1&1&0\\
0&1&0\\
0&0&1
\end{array}
\right ]~~~
b=\left [
\begin{array}{ccc}
1&0&0\\
0&1&1\\
0&0&1
\end{array}
\right ]~~~
c=\left [
\begin{array}{ccc}
1&0&1\\
0&1&0\\
0&0&1
\end{array}
\right ]
$$

Any element $g \in \mathbf{H}$ can be written uniquely as $b^ja^ic^k$.
$$ g=\left [
\begin{array}{ccc}
1&i&k\\
0&1&j\\
0&0&1
\end{array}
\right ] = b^ja^ic^k
$$

It is shown in \cite{Re10} that the languages
$\mathtt{MULT}=\{x^py^qz^{pq} | p,q\geq 0\}$,
$\mathtt{COMPOSITE}=\{x^{pq}| p,q >1\}$ and
$\mathtt{MULTIPLE}=\{x^py^{pn}|p \in \mathbb{N}\}$ can be recognized
by a $\mathbf{H}$-automaton, using the special multiplication property
of the group.

Correcting a small error in \cite{Re10}, we rewrite the multiplication
property of the elements
of $\mathbf{H}$.

\begin{equation*}\label{eq0}
 (b^xa^yc^z)(b^{x'}a^{y'}c^{z'})= b^{x+x'}a^{y+y'}c^{z+z'+yx'} 
\end{equation*}

We can make the following observation using the fact that
$\mathfrak{L}(\mathbf{H})$ contains non-context-free languages.

\begin{theorem}
	$\mathfrak{L}(SL(2,\mathbb{Z})) \subsetneq \mathfrak{L}(SL(3,\mathbb{Z}))$.
\end{theorem}
\begin{proof}
	It is obvious that a $SL(2,\mathbb{Z})$-automaton can be simulated by
	a $SL(3,\mathbb{Z})$-automaton.
	Note that $\mathfrak{L}(SL(2,\mathbb{Z}))$ is the family of
	context-free languages by Theorem \ref{theorem:gl}. Since
	$\mathfrak{L}(\mathbf{H}) \subsetneq \mathfrak{L}(SL(3,\mathbb{Z}))$
	and the non-context-free language $\mathtt{COMPOSITE}=\{x^{pq}| p,q
	>1\}$ can be recognized by a $ \mathbf{H} $-automaton \cite{Re10}, the
	result follows.
\end{proof}

Now let us move on to the discussion about matrix groups with rational entries.

Let us denote by $GL(2,\mathbb{Q})$ the general linear group of degree
two over the field of rational  numbers, that is, the group of
invertible matrices with rational entries.  Restricting the matrices
in $GL(2,\mathbb{Q})$ to those that have determinant 1, we obtain the
special  linear group of degree two over the field of rationals,
$SL(2,\mathbb{Q})$.

We will start by proving that allowing rational entries enlarges the
class of languages recognized by matrices with determinant 1.

\begin{theorem}\label{theorem: UPOWODD}
$ \mathfrak{L}(SL(2,\mathbb{Z})) \subsetneq \mathfrak{L}(SL(2,\mathbb{Q}))  $.
\end{theorem}
\begin{proof} It is obvious that $ \mathfrak{L}(SL(2,\mathbb{Z}))
	\subseteq \mathfrak{L}(SL(2,\mathbb{Q}))  $. We will prove that the
	inclusion is proper.
	
	Let us construct a $SL(2,\mathbb{Q})$-automaton $\mathcal{G}$
	recognizing the language $\mathtt{L}=\{a^{2^{2n+1}} | n \geq 0 \}$. 
	The state diagram of $ \mathcal{G} $ and the matrices are given in 
	Figure \ref{fig:uoddpow}. Without scanning
	any input symbol, $\mathcal{G}$ first multiplies its register with
	the matrix $A_1$. $\mathcal{G}$ then multiplies its register with the
	matrix $ A_2 $ successively until nondeterministically moving to the
	next state. After that point, $\mathcal{G}$ starts reading the string
	and multiplies its register with the matrix $A_3$ for each scanned
	$a$. At some point, $\mathcal{G}$ nondeterministically stops reading
	the rest of the string and multiplies its register with the matrix
	$A_4$. After sucessive multiplications with $A_4$, $\mathcal{G}$
	nondeterministically decides moving to an accept state. 
	
	\begin{figure}[h]
		\centering
		\includegraphics[width=0.8\linewidth]{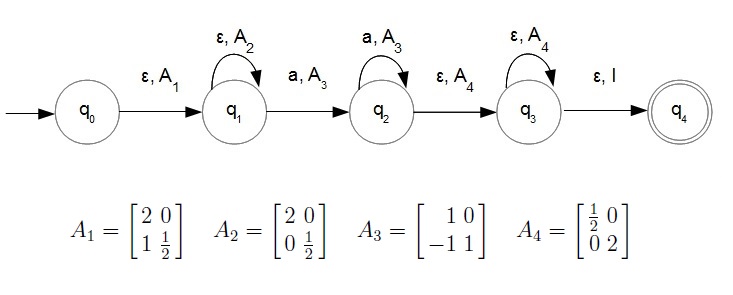}
		\caption{State diagram of $ \mathcal{G} $ accepting the language 
		$\mathtt{L}=\{a^{2^{2n+1}} | n \geq 0 \}$ }
		\label{fig:uoddpow}	
	\end{figure}

	Let us trace the value of the register at different stages of the
	computation. Before reading the first input symbol, the register has
	the value
	$$\left [
	\begin{array}{cc}
	2^{x+1}&0\\
	2^x  &\frac{1}{2^{x+1}}\\
	\end{array}
	\right ]$$
	
	\noindent as a result of the multiplications with the matrix $A_1$
	and $ x $ many $ A_2 $'s. Multiplication
	with each $A_3$ leaves $2^{x+1}$ and $ \frac{1}{2^{x+1}} $ unchanged
	while subtracting $ \frac{1}{2^{x+1}} $ from $2^x$ for each scanned
	$a$. As a result of $ y $  multiplications with $ A_3 $, the register
	will have the value
	
	$$\left [
	\begin{array}{cc}
	2^{x+1}&0\\
	2^x-\frac{y}{2^{x+1}} &\frac{1}{2^{x+1}}\\
	\end{array}
	\right ].$$
	
	For the rest of the computation, $\mathcal{G}$ will multiply its
	register with $A_4$ until
	nondeterministically moving to the final state. As a result of $ z $
	multiplications with $ A_4 $, the register will have the value
	
	$$\left [
	\begin{array}{cc}
	\frac{2^{x+1}}{2^z}&0\\
	\bigl( 2^x-\frac{y}{2^{x+1}}\bigr)\frac{1}{2^z} &\frac{2^z}{2^{x+1}}\\
	\end{array}
	\right ].$$
	
	The final value of the register is equal to the identity matrix when
	$ y=2^{2x+1} $ and $ z=x+1 $, which is possible only when the length
	of the input string is $ 2^{2x+1} $ for some $ x\geq 0 $. In the
	successful branch, the register will be equal to the identity matrix
	and $\mathcal{G}$ will end up in the final state having successfully
	read the input string. For input strings which are not members of
	$\mathtt{L}$, either the computation will end before reading the whole
	input string, or the final state will be reached with the register
	value not equaling the identity matrix.
	
	Since the matrices used during the computation are 2 by 2 rational
	matrices with determinant 1, $ \mathtt{L} \in
	\mathfrak{L}(SL(2,\mathbb{Q})) $. $ \mathfrak{L}(SL(2,\mathbb{Q})) $
	contains a unary nonregular language, which is not true for $
	\mathfrak{L}(SL(2,\mathbb{Z})) $ by Theorem \ref{theorem:gl} and we
	conclude the result.
\end{proof}

We will now look at a special subgroup of $GL(2,\mathbb{Q})$.

For two integers $m$ and $n$, the \textit{Baumslag-Solitar group}
$BS(m,n)$ is defined as $BS(m,n)=\langle
a,b | ba^mb^{-1}=a^n\rangle$. We are going to focus on
$BS(1,2)=\langle a,b|bab^{-1}=a^2\rangle$.

Consider the matrix group $G_{BS}$ generated by the matrices

$$A=\left [
\begin{array}{cc}
1&0\\
-1 &1\\
\end{array}
\right ]~~~\mbox{and}~~~
B=\left [
\begin{array}{cc}
1/2&0\\
0 &1\\
\end{array}
\right ].$$

Consider the isomorphism $a \mapsto A$, $b \mapsto B$. The matrices
$A$ and $B$ satisfy the property $BAB^{-1} = A^2$,
\[ \left [
\begin{array}{cc}
1/2&0\\
0 &1\\
\end{array}
\right ]
\left [
\begin{array}{cc}
1&0\\
-1 &1\\
\end{array}
\right ]
\left [
\begin{array}{cc}
2&0\\
0&1\\
\end{array}
\right ]=
\left [
\begin{array}{cc}
1&0\\
-2&1\\
\end{array}
\right ],  \]
and we conclude that $G_{BS}  $ is isomorphic to $BS(1,2)$.

We will prove that there exists a $ BS(1,2) $-automaton which
recognizes a non-context-free language.

\begin{theorem}\label{theorem: upow}
	$ \mathfrak{L}(BS(1,2)) \nsubseteq \mathsf{CF} $.
\end{theorem}
\begin{proof}
	Let us construct a $BS(1,2)$-automaton $\mathcal{G}$ recognizing the
	language $\mathtt{UPOW}=\{a^{2^n}|n\geq 0\}$. The state diagram of 
	$ \mathcal{G} $ and the matrices are given in Figure \ref{fig:upow}. 
	Without scanning any input symbol, $\mathcal{G}$ multiplies its register 
	with the matrix $A_1$ successively. $\mathcal{G}$
	nondeterministically moves to the next state reading the first input
	symbol without modifying the register.
	After that point, $\mathcal{G}$ starts reading the string and
	multiplies its register with the matrix $A_2$
	for each scanned $a$. At some point, $\mathcal{G}$
	nondeterministically stops reading
	the rest of the string and multiplies its register with the element
	$A_3$. After successive multiplications
	with $A_3$, $\mathcal{G}$ nondeterministically decides to move to an
	accept state. 
	\begin{figure}[!htb]
		\centering
		\includegraphics[width=0.6\linewidth]{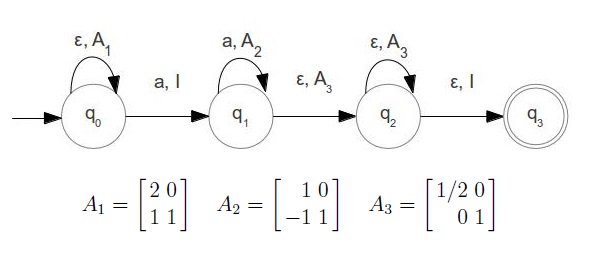}
		\caption{State diagram of $ \mathcal{G} $ recognizing
		\label{fig:upow}
		 $\mathtt{UPOW}=\{a^{2^n}|n\geq 0\}$}
	\end{figure}

	Before reading the first input symbol, the register has the value
	$$\left [
	\begin{array}{cc}
	2^k&0\\
	2^k -1 &1\\
	\end{array}
	\right ]$$
	for some $k\geq 0$ as a result of the multiplications with the matrix
	$A_1$. Multiplication
	with each $A_2$ leaves $2^k$ unchanged while subtracting 1 from $2^k -
	1$ for each scanned $a$. For a member
	input string of the form $a^{2^k}$, in the successful branch
	$\mathcal{G}$ will multiply its register with
	$A_2$ until reaching the end of the string and the register will have the value
	$$\left [
	\begin{array}{cc}
	2^k&0\\
	0 &1\\
	\end{array}
	\right ].$$
	For the rest of the computation, $\mathcal{G}$ will multiply its
	register with $A_3$ until
	nondeterministically moving to the final state. In the successful
	branch, the register will be equal to the
	identity matrix and $\mathcal{G}$ will end up in the final state
	having successfully read the input string.
	For input strings which are not members of $\mathtt{UPOW}$, either the
	computation will end before reading the
	whole input string or the final state will be reached with the
	register value being different from the
	identity matrix.
	Note that $A_1=B^{-1}A^{-1}$, $A_2=A$ and $A_3=B$, where $A$ and $B$
	are the generators of the group $G_{BS}$ and recall that $G_{BS}$ is
	isomorphic to $BS(1,2)$. Since $ \mathtt{UPOW} $ is a unary nonregular
	language, it is not context-free and we conclude the result.
\end{proof}
\begin{figure}[hbt]
	\centering
	\includegraphics[width=0.9\linewidth]{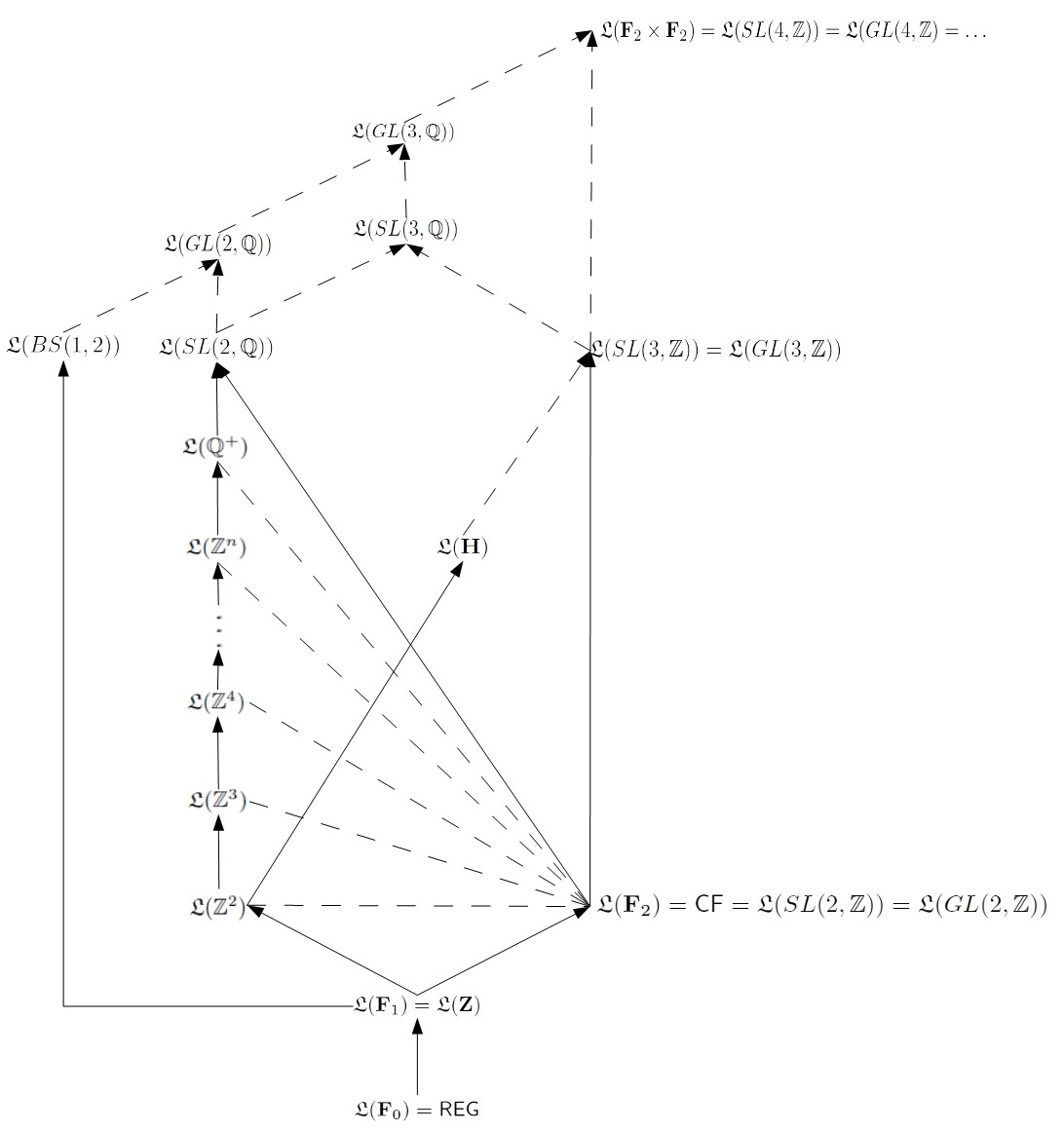}
	\caption{Language classes associated with groups}
	\label{fig: diagram}
\end{figure}

Note that $\mathfrak{L}(\mathbb{Z}) \subsetneq \mathfrak{L}(BS(1,2)) $ 
since the subgroup generated by $ a $ in $ BS(1,2) $ is isomorphic to 
$ \mathbb{Z} $ and $\mathfrak{L} (BS(1,2)) $ contains a unary nonregular language.

We summarize the results in Figure \ref{fig: diagram}. Solid arrows
represent proper inclusion, dashed arrows represent inclusion and
dashed lines represent incomparability. For the relationships which 
are not discussed in this section, please see Section~\ref{Section: appendix}.

\section{Time Complexity}
\label{Section: time}

A group automaton $ \mathcal{G} $ recognizing language $\mathtt{L} $
is said to  \textit{operate in time} $ t(n) $, if for any input string
$ x $ with $ |x|=n $ the computation of $ \mathcal{G} $ takes at most
$ t(n) $ steps. We will denote the set of languages recognized by $ G
$-automata operating in time $ t(n) $ by $ \mathfrak{L}(G)_{t(n)} $.

Let $ X $ be a generator set for $ G $. The \textit{length} of $ g \in
G $, denoted $|g|_X$, is the length of the shortest representative for $
g $ in $ X^* $. The \textit{growth function of a group} $ G $ with
respect to a generating set $ X $, denoted $ g^{X}_G(n)$, is the cardinality 
of the set $ \{g \in G,|g|_X\leq n \} $, that is the number of all elements in 
$ G $ which can be represented by a word of length at most $ n $. 
The growth function is asymptotically independent of the generating set, 
and we will denote the growth function of a group $ G $ by $ g_G(n)$.

For a positive integer $ n $, two strings $ w,w' \in \Sigma^* $ are $
n $-\textit{dissimilar}, if $ |w|\leq n $, $ |w'|\leq n $ and there
exists a distinguishing string $ v \in \Sigma^*  $ with $ |wv|\leq n
$, $ |w'v|\leq n $ and $ wv \in \mathtt{L} $ iff $w'v \notin
\mathtt{L}  $. Let $ N_\mathtt{L}(n) $ be the maximum $ k $ such that
there exist $ k $ distinct strings that are pairwise $n $-dissimilar.
For a string $ w=w_1w_2\dots w_k \in W(G)$, $ w^{-1} $ represents the
word $w_k^{-1}w_{k-1}^{-1}\dots w_1^{-1}  $.

\begin{lemma}\label{lemma: growth}
	Let $ G $ be a finitely generated group with growth function $ g_G(n)
	$. $ N_{W(G)}(n)\geq g_G(\frac{n}{2})$.
\end{lemma}
\begin{proof}
	Let $ X $ be the generator set of $ G $. The strings in $ W(G) $ are
	those which belong to $ (X \cup X^{-1})^* $ and represent the identity
	element of $ G $. Let $ w=w_1\dots w_k $ be a word of length less than
	or equal to $\frac{n}{2} $. Each word of length $ k $ can be extended
	with $w^{-1}= w_k^{-1}\dots w_1^{-1} $ so that the extended word
	represents the identity element of $ G $ and the word has length less
	than or equal to $n $. The number of distinct elements $ g $ in $ G $
	which can be represented by a word of length less than or equal to
	$\frac{n}{2} $ is $ g_G(\frac{n}{2})$, which is the cardinality of the set 
	$\{g \in G,|g|_X\leq \frac{n}{2}\}$. Since the set 
	$ \{g \in G,|g|_X\leq \frac{n}{2}\} $  contains only the shortest 
	representative of each element, any two elements $ g_1, g_2 $ are 
	distinct such that the string $ g_1g_1^{-1} \in W(G)$ whereas the string 
	$ g_2g_1^{-1} \notin W(G)$. We conclude that there are at least 
	$  g_G(\frac{n}{2}) $ distinct strings that are pairwise $n $-dissimilar.
\end{proof}

\begin{theorem}\label{theorem: growth}
	Let $ G $ and $ H $ be groups with polynomial and exponential growth
	functions $ g_G(n)$ and $  g_H(n) $ respectively. $\mathfrak{L}(H)
	\nsubseteq \mathfrak{L}(G)_{t(n)}  $ where $ t(n) $ is a polynomial.
\end{theorem}
\begin{proof}
	Let $ \mathcal{G} $ be a $ G $-automaton recognizing a language $ L $
	in time $ t(n) $.  A configuration of a group automaton consists of a state
	and a group element pair. Let us count the number of distinct configurations
	of $ \mathcal{G} $ that can be reached after reading a string of
	length at most $ m $. Since the number of states is constant, we
	will analyze the number of different group elements that can appear in
	the register. After reading a string of length exactly $ m $, the product of
	the labels on the edges can be given by
	$$ l=g_{i_1}g_{i_2}\dots g_{i_{k}} $$ for some $ k \leq t(m) $. $ l $ can be 
	expressed as a product of $ \kappa $ generators, where $ \kappa $ is at 
	most $C\cdot k $ for some constant $ C $, since each group
	element is composed of at most some constant number of
	generators, which is independent of the length of the string. The number of 
	elements in $ G $ which can be represented as a product of at most $ \kappa$ 
	generators is given by $ g_G(\kappa) $ by the definition of the growth function 
	of $ G $. Hence, the number of different values that can appear in the register 
	after reading a string of length exactly $ m $ is less than or equal to 
	$ g_G(\kappa) $. Since  $ \kappa \leq  C\cdot k $ and $ k \leq t(m) $, we can 
	conclude that $$ g_G(\kappa)  \leq  g_G(C \cdot t(m)). $$
	
	$ g_G(C \cdot t(m)) $ is a polynomial function of $ m $ since both the growth 
	function of $ G $ and $ t(m) $ have polynomial growth. Now it is easy to see 
	that the number of different group elements that can appear in the register after
	reading a string of length \textit{at most} $ m $, it is still a polynomial function of
	$ m $.  
	
	Now let us consider the strings in $ W(H) $. $g_H(\frac{n}{2}) \leq
	N_{W(H)}(n)$ by Lemma \ref{lemma: growth}. Hence, there are at least $
	g_H(\frac{n}{2}) $ many distinct $ x_i ,x_j$ such that $ x_iy
	\in W(H)$ whereas $ x_jy \notin W(H) $ and $ |x_iy|\leq
	n $ and $ |x_jy|\leq n $ for some $ y $ . Let us consider the accepting
	computation paths for these strings. After finishing reading the $ x_i
	$, a configuration should be reached which eventually leads to an
	accept state. Since the total number of distinct configurations after
	reading a prefix of length at most $ n$ is a polynomial function of $
	n $, the number of configurations which leads to an accept state is
	also a polynomial function of $ n $. We can conclude that the same
	configuration should be reached after reading two distinct strings $
	x_i $ and $ x_j $ since there are at least $ g_H(\frac{n}{2})  $ many
	such different strings which is an exponential function of $ n $. We
	have assumed that the configuration in consideration leads to an
	accept state, that is an accepting configuration is reached if the
	string $ x_i $ is extended with $ y $. This will result in the
	acceptance of the string $ x_jy $, which is not a member of $
	W(H) $ since otherwise $x_i $ would be equivalent to  $x_j $. We
	arrive at a contradiction and conclude that $W(H) $ cannot be
	recognized by a $ G $-automaton in polynomial time.
	
	We conclude that $\mathfrak{L}(H) \nsubseteq \mathfrak{L}(G)_{t(n)}  $
	since $ W(H) $ is trivially in $ \mathfrak{L}(H) $.
\end{proof}

\begin{theorem}\label{thm: polycf}
	Let $ G $ be a group with a polynomial growth function. There exists a
	context-free language which cannot be recognized by any $ G
	$-automaton in polynomial time.
\end{theorem}
\begin{proof}
	It is known that the word problem of the free group of rank $
	W(\mathbf{F}_2) $ has an exponential growth function. Assuming that  $
	G $ is a group with polynomial growth function, $ W(\mathbf{F}_2) $
	cannot be recognized by any $ G $-automaton in polynomial time by
	Theorem \ref{theorem: growth}. Since $ W(\mathbf{F}_2) $ is a
	context-free language, the proof is complete.
\end{proof}
\begin{corollary}
	$ \mathsf{CF} \nsubseteq \mathfrak{L}(\mathbf{\mathbf{H}})_{Poly} $.
\end{corollary}
\begin{proof}
	It is known that the Discrete Heisenberg group $ \mathbf{H}$ has
	polynomial growth function. The result follows by Theorem \ref{thm: polycf}.
\end{proof}

\section{Additional Results} \label{Section: appendix}
	
	In this section, we are going to state some known results from the
	literature and some new results which help us complete Figure
	\ref{fig: diagram}.
	
	The relation between the classes of languages recognized by free group
	automata is summarized as follows.
	
\begin{fact}\textup{\cite{DM00}}		
		$ \mathsf{REG} = \mathfrak{L}(\mathbf{F}_0) \subsetneq
		\mathfrak{L}(\mathbf{F}_1) = \mathfrak{L}(\mathbb{Z}) \subsetneq
		\mathfrak{L}(\mathbf{F}_2) $.
\end{fact}

The following result states the hierarchy between the classes of languages 
recognized by $ \mathbb{Z}^n $-automata. This result also follows from the
hierarchy between the class of languages recognized by nondeterministic 
blind $ k $-counter automata.

	\begin{fact}
		$ \mathfrak{L}(\mathbb{Z}^k) \subsetneq
		\mathfrak{L}(\mathbb{Z}^{k+1})$ for $ k \geq 1 $.
	\end{fact}
	
	At the top of the hierarchy of $ \mathbb{Z}^k $, we place $
	\mathbb{Q}^+ $, which is isomorphic to a free Abelian group of
	infinite rank. Let us note that the set of languages recognized by $
	\mathbb{Q}^+ $-automata is a proper subset of the set of languages
	recognized by $ SL(2,\mathbb{Q}) $-automata which can be concluded 
	with the help of the following fact.
	
	\begin{fact}\textup{\cite{ISK76}}\label{fact: unary}
		All \textup{1NFAMW}-recognizable languages over a unary alphabet are regular.
	\end{fact}
	
	\begin{theorem}
		$ \mathfrak{L}(\mathbb{Q}^+) \subsetneq \mathfrak{L}(SL(2,\mathbb{Q}))  $.
	\end{theorem}
	
	\begin{proof} Let $ \mathtt{L} \in \mathfrak{L}(\mathbb{Q}^+) $ and let $
		\mathcal{G} $ be a $ \mathbb{Q}^+ $-automaton recognizing $ \mathtt{L}
		$. We will construct a $ \mathfrak{L}(SL(2,\mathbb{Q})) $-automaton  $
		\mathcal{G}' $ recognizing $ \mathtt{L} $. Let $ S =\{s_1,\dots,s_n\}
		$ be the set of elements multiplied with the register during the
		computation of $ \mathcal{G} $. We define the mapping $ \varphi $ as
		follows.
		\[ \varphi: s_i
		\mapsto
		\left [
		\begin{array}{cc}
		s_i&0\\
		0&\frac{1}{s_i}\\
		\end{array}
		\right ]~~~
		\]
		The elements $\varphi(s_i) $ are $ 2 \times 2 $ rational matrices with
		determinant 1. Let $ \delta $ and $ \delta' $ be the transition
		functions of $ \mathcal{G} $ and $ \mathcal{G}' $ respectively. We let
		$ (q',s_i) \in \delta(q,\sigma) \iff (q',\varphi(s_i)) \in
		\delta'(q,\sigma)$ for every $ q,q' \in Q$, $ \sigma \in \Sigma $ and
		$ s_i \in S $. The resulting $ \mathcal{G}' $ recognizes $ \mathtt{L}$.
		
		The inclusion is proper since  $\mathtt{L}=\{a^{2^{2n+1}} | n \geq 0
		\} \in \mathfrak{L}(SL(2,\mathbb{Q}))$ by Theorem \ref{theorem:
		UPOWODD}, and $ \mathfrak{L}(\mathbb{Q}^+) $ does not contain any
		unary nonregular languages by Fact \ref{fact: unary}, noting that 
		$ \mathbb{Q}^+ $-automata are equivalent to 1NFAMW's.
	\end{proof}
	
	Let us mention that the class of context-free languages and the class
	of languages recognized by nondeterministic blind counter automata are
	incomparable.
	
	\begin{fact}
		$\mathsf{CF}$ and $\mathfrak{L}(\mathbb{Z}^n)$ are incomparable for
		all $ n \geq 2 $.		
	\end{fact}
	
	\begin{proof} Consider the language $\mathtt{L}=\{a^nb^n| n\geq 0\}$ which is a
		context-free language. Since context-free languages are closed under star, 
		$\mathtt{L}^*$ is a context-free language whereas it cannot be recognized 
		by any $\mathbb{Z}^n$-automaton for all $ n \geq 1 $ by
		\cite{Gr78}. On the other hand, the non-context-free language 
		$\mathtt{L}'=\{a^nb^nc^n|n \geq 0\}$ can be recognized by a
		$\mathbb{Z}^2$-automaton.
	\end{proof}
	
	Let us move on to the results about linear groups. The following result
	is a direct consequence of Fact \ref{fact: finite}.
		
	\begin{theorem}	
		$ \mathfrak{L}(SL(3,\mathbb{Z}))=\mathfrak{L}(GL(3,\mathbb{Z})) $.
	\end{theorem}
	
	\begin{proof} Since $ GL(3,\mathbb{Z}) $ is a finitely generated group and $
		SL(3,\mathbb{Z}) $ has finite index in $ GL(3,\mathbb{Z}) $, the
		result follows by Fact \ref{fact: finite}.
	\end{proof}
	
	We have talked about the discrete Heisenberg group $ \textbf{H} $, an
	important subgroup of $ SL(3,\mathbb{Z}) $. Now let us look at a
	subgroup of $\mathbf{H}$ generated by the matrices $B$ and $C$ which
	we will call $\mathbf{G_2}$.
	
	$$ B=\left [	
	\begin{array}{ccc}
	1&0&0\\
		0&1&1\\
		0&0&1
		\end{array}
		\right ]~~~
		C=\left [
		\begin{array}{ccc}
		1&0&1\\
		0&1&0\\
		0&0&1
		\end{array}
		\right ]~~~
		$$	
		$\mathbf{G_2}=\langle B,C |BC=CB\rangle $ is a free Abelian group of
	rank 2 and therefore it is isomorphic to $\mathbb{Z}^2$.
	
	We conclude the following about the language recognition power of $
	\mathbb{Z}^2$ and $\mathbf{H}$.
		
	\begin{theorem}\label{theorem: Z2H}
		$\mathfrak{L}(\mathbb{Z}^2) \subsetneq \mathfrak{L}(\mathbf{H})$.
	\end{theorem}
	\begin{proof}
		Since $\mathbb{Z}^2$ is a subgroup of $\mathbf{H}$,
		$\mathfrak{L}(\mathbb{Z}^2) \subseteq 
		\mathfrak{L}(\mathbf{H}) $ follows. The inclusion is proper since
		$\mathbf{H}$ can recognize the unary nonregular language 
		$\mathtt{COMPOSITE}=\{x^{pq}| p,q >1\}$ by
		\cite{Re10}, which is not possible for any $\mathbb{Z}^n$-automaton by
		Fact \ref{fact: unary}.
	\end{proof}
	
	In \textup{\cite{MS01}}, it is proven that $ \mathbf{F}_2 \times
	\mathbf{F}_2 $-automaton is as powerful as a Turing machine, which
	places  $ \mathbf{F}_2 \times \mathbf{F}_2 $ at the top of the
	language hierarchy of group automata.
	
	\begin{fact}\textup{\cite{MS01}}\label{fact: RE}		
		$ \mathfrak{L} ( \mathbf{F}_2 \times \mathbf{F}_2 )$ is the family of
		recursively enumerable languages.
	\end{fact}
	
	We make the following observation.
	
	\begin{theorem}
		$\mathsf{RE} =  \mathfrak{L} ( \mathbf{F}_2 \times \mathbf{F}_2 ) =
		\mathfrak{L} (SL(n,\mathbb{Z}) )$ for $ n \geq 4 $.
	\end{theorem}
	
	\begin{proof} The first equality is Fact \ref{fact: RE}.
		Recall that $ \varphi $ is an isomorphism from $ \mathbf{F}_2$ onto $
		\mathbf{G} $, the matrix group generated by the matrices $ M_a $ and $
		M_b $.
		Let $ \mathbf{G}' $ be the following group of matrices
		\[ \left \{
		\left [
		\begin{array}{clll}
		\multicolumn{2}{l}
		{\multirow{2}{*}{$M_1$}} & 0 & 0 \\
			& & 0 & 0 \\
		0&0 & \multicolumn{2}{c}{\multirow{2}{*}{ $M_2$}} \\
		0&0 & & \\
		\end{array}
		\right ], \ M_1, \ M_2 \in \mathbf{G} \right \}
		.\]
		
		We will define the mapping $ \psi:\mathbf{F}_2 \times \mathbf{F}_2
		\rightarrow \mathbf{G}' $ as $ \psi(g_1,g_2) =
		(\varphi(g_1),\varphi(g_2)) $ for all $ (g_1,g_2)\in \mathbf{F}_2
		\times \mathbf{F}_2 $ which is an isomorphism from $\mathbf{F}_2
		\times \mathbf{F}_2  $ onto $ \mathbf{G}' $.
		
		This proves that $\mathbf{F}_2 \times \mathbf{F}_2  $ is isomorphic to
		a subgroup of $ SL(4,\mathbb{Z}) $. The fact that $  \mathfrak{L} (
		\mathbf{F}_2 \times \mathbf{F}_2 ) $ is the set of recursively
		enumerable languages helps us to conclude that $
		\mathfrak{L}(SL(n,\mathbb{Z})) $  is the set of recursively enumerable
		languages for $ n \geq 4 $.
		\end{proof}
	
	Let us also state that the classes of languages recognized by automata
	over supergroups of $ SL(4,\mathbb{Z}) $ such as $ GL(4,\mathbb{Z})
	$ or $ SL(4,\mathbb{Q}) $ are also identical to the class of
	recursively enumerable languages, since such automata can be simulated
	by Turing machines.

\section{Open Questions}\label{Section: open}

Does there exist a $ SL(3,\mathbb{Z}) $-automaton recognizing $ W(\mathbb{Z}^3) $?
\footnote{Corollary 2 of \cite{CEO06} states that the word problem of a finitely generated 
Abelian group $H$ is recognized by a $G$-automaton if and only if $H$ has a finite index 
subgroup isomorphic to a subgroup of $G$. That corollary could be used to give an 
affirmative answer to this open question. Unfortunately, the corollary is wrong: 
Let $ H $ be an Abelian group and let $ G=\mathbf{F}_2 \times \mathbf{F}_2 $. 
$ \mathfrak{L}(\mathbf{F}_2 \times \mathbf{F}_2 )$ contains the word problem of any 
finitely generated Abelian group. Since $ \mathbf{F}_2 \times \mathbf{F}_2 $ is finitely 
generated, any finite index subgroup of $ \mathbf{F}_2 \times \mathbf{F}_2 $ is also 
finitely generated. Any finite index subgroup of $ \mathbf{F}_2 \times \mathbf{F}_2 $ is 
either free or has a subgroup of finite index that is a direct product of free groups 
\cite{BR84}. Any subgroup of an Abelian group is again Abelian. Hence, it is not 
possible that $ G$ has a finite index subgroup isomorphic to a subgroup of $ H $.}

Can we prove a stronger version of Theorem \ref{thm: polycf}, which is independent 
of the time component? For instance, for the case of $ \mathbf{F}_2 $, is it true that 
$ W(\mathbf{F}_2) \notin \mathfrak{L}(\mathbf{H})$ in general?

Can we describe the necessary properties of a group $ G $ so that 
$ \mathfrak{L}(G) $ contains $ W(\textbf{F}_2) $?

Little is known about $ BS(1,2) $-automata. Does $  \mathfrak{L}(BS(1,2)) $ 
contain every context-free language?

Which, if any, of the subset relationships in Figure \ref{fig: diagram} are 
proper inclusions?

\section*{Acknowledgements}
We thank the anonymous reviewers for their constructive comments.

\biblio{references}

\EndOfArticle 

\end{document}